\documentclass[submission,copyright,creativecommons]{eptcs}
\providecommand{\event}{AFL 2014} % Name of the event you are submitting to
\usepackage{breakurl}             % Not needed if you use pdflatex only.

\usepackage{amssymb, amsmath, amsthm, amstext}
\usepackage{enumitem}
\usepackage{url}
\usepackage{color}

\usepackage{float}
\usepackage{graphicx}
\usepackage{epstopdf}

\allowdisplaybreaks[1]

\newcommand{\coloneq}{\mathrel{\mathop:}=}
\newcommand{\intersect}{\cap}
\newcommand{\union}{\cup}

\newcommand{\bigunion}{\bigcup}

\newcommand{\delete}{\setminus}
\newcommand{\ie}{i.e.}
\newcommand{\st}{\;:\;}

\newcommand{\superscript}[1]{\ensuremath{^{\textrm{#1}}}}

\newcommand{\dotdot}{\ldotp\ldotp}
\newcommand{\bbN}{\mathbb{N}}

\newcommand{\calO}{\mathcal{O}}

\providecommand{\abs}[1]{\left\lvert#1\right\rvert}

\DeclareMathOperator{\LSD}{LSD}
\DeclareMathOperator{\USD}{USD}
\DeclareMathOperator{\SD}{SD}
\DeclareMathOperator{\FBE}{FBE}

\begin{document}

\title{Similarity density of the Thue-Morse word with overlap-free infinite binary words}
\author{Chen Fei Du and Jeffrey Shallit 
\institute{School of Computer Science,
University of Waterloo,
Waterloo, ON N2L 3G1, Canada}
\email{cfdu@uwaterloo.ca}, \ \email{shallit@uwaterloo.ca}}
\def\titlerunning{Similarity Density}
\def\authorrunning{C. F. Du and J. Shallit}

\maketitle

\theoremstyle{plain}
\newtheorem{theorem}{Theorem}
\newtheorem{thm}[theorem]{Theorem}
\newtheorem{corollary}[theorem]{Corollary}
\newtheorem{cor}[theorem]{Corollary}
\newtheorem{lemma}[theorem]{Lemma}
\newtheorem{lem}[theorem]{Lemma}
\newtheorem{proposition}[theorem]{Proposition}
\newtheorem{prop}[theorem]{Proposition}
\newtheorem{observation}[theorem]{Observation}
\newtheorem{obsv}[theorem]{Observation}

\theoremstyle{definition}
\newtheorem{definition}[theorem]{Definition}
\newtheorem{defn}[theorem]{Definition}
\newtheorem{example}[theorem]{Example}
\newtheorem{eg}[theorem]{Example}
\newtheorem{problem}[theorem]{Problem}
\newtheorem{prob}[theorem]{Problem}
\newtheorem{question}[theorem]{Question}
\newtheorem{quest}[theorem]{Question}
\newtheorem{conjecture}[theorem]{Conjecture}
\newtheorem{conj}[theorem]{Conjecture}

\theoremstyle{remark}
\newtheorem{remark}[theorem]{Remark}
\newtheorem{rem}[theorem]{Remark}

\begin{abstract}
We consider a measure of similarity for infinite words that generalizes the notion of asymptotic or natural density of subsets of natural numbers from number theory.
We show that every overlap-free infinite binary word, other than the Thue-Morse word ${\bf t}$ and its complement $\overline{\bf t}$, has this measure of similarity with ${\bf t}$ between $\frac{1}{4}$ and $\frac{3}{4}$.
This is a partial generalization of a classical 1927 result of Mahler.
\end{abstract}

\section{Introduction}

The Thue-Morse word
\[
{\bf t} = {\tt 01101001100101101001011001101001} \cdots
\]
is one of the most studied objects in combinatorics on words. It can be defined in a number of different ways, such as the fixed point of the morphism $\mu$ defined by $\mu({\tt 0}) \coloneq {\tt 01}$ and $\mu({\tt 1}) \coloneq {\tt 10}$ beginning with ${\tt 0}$, or as the word whose $n$\superscript{th} position is the number of {\tt 1}s (modulo 2) in the binary representation of $n$.

The word $\bf t$ has a large number of interesting properties, many of which are covered in the survey \cite{Allouche&Shallit:1999}. For example, $\bf t$ is {\it overlap-free}: it contains no factor of the form $axaxa$, where $x$ is a (possibly empty) word and $a$ is a single letter.
One that concerns us here is the following ``fragility'' property \cite{Brown&Rampersad&Shallit&Vasiga:2006}: if the bits in any {\it finite} non-empty set of positions are ``flipped'' (\ie, changed to their binary complement) in the 
Thue-Morse word,
the resulting word is no longer overlap-free.\footnote{Note
that the ``fragility'' property does not hold for an arbitrary overlap-free binary word; for example, both ${\tt 0} {\bf t}$ and ${\tt 1} {\bf t}$ are overlap-free. There are even overlap-free words in which blocks arbitrarily far from the beginning may be flipped and still remain 
overlap-free \cite{Shallit:2011}.}

Of course, this is not true of arbitrary {\it infinite} sets of positions; for example, we can transform $\bf t$ to $\overline{\bf t}$ by flipping {\it all} the positions.
Chao Hsien Lin (personal communication, October 2013) raised the following natural question.

\begin{prob} \label{prob:1}
Is it possible to flip an {\it infinite}, but density $0$, set of positions in {\bf t} and still get an overlap-free word?
\end{prob}

Our main result (Theorem~\ref{thm:main}) solves Problem~\ref{prob:1} in the negative.
After making precise what we mean by ``density'', we use a
certain automaton \cite{Shallit:2011} encoding all the
overlap-free infinite binary words to compare ${\bf t}$ to all other overlap-free infinite binary words and show that they differ from ${\bf t}$ in at least density $\frac{1}{4}$ of the positions.
Furthermore, computational evidence suggests
that the true lower bound is density $\frac{1}{3}$.
However, we were unable to obtain a proof of this tighter bound.
Finally, we consider the possibility of similar results holding for other words (in place of ${\bf t}$) or for larger classes of words (in place of overlap-free words).

\section{Notation}

We observe the following notational conventions throughout this paper.
We let $\bbN \coloneq \{ 0,1,2, \dots \}$ denote the natural numbers.
The upper-case Greek letters $\Sigma, \Delta, \Gamma$ represent finite alphabets.
For each $n \in \bbN$, we let $\Sigma_n \coloneq \{ 0, 1, 2, \dots, n-1 \}$.

As usual, $\Sigma^\omega$ denotes the set of all (right-)infinite words over $\Sigma$ and $L^\omega \coloneq \{ x_0 x_1 x_2 \cdots \ \st \ x_i \in L \delete \{ \epsilon \} \}$ denote the set of all infinite words formed by concatenation from nonempty words of $L$.  By $x^\omega$ we mean the infinite periodic word $xxx \cdots$.

We adopt the convention that, in the context of words, lower-case letters such as $x,y,z$ refer to finite words (\ie, $x,y,z \in \Sigma^*$), while boldface letters ${\bf x}, {\bf y}, {\bf z}$ refer to infinite words (\ie, ${\bf x}, {\bf y}, {\bf z} \in \Sigma^\omega$).

To be consistent with $0 \in \bbN$, all words are zero-indexed, \ie, the first letter of the word is in position $0$.
For $x \in \Sigma^*$ and $m \leq n \in
\bbN$, $x[n]$ denotes the letter at the $n$\superscript{th} position of $x$ and $x[m\dotdot n]$ denotes the subword consisting of the letters from the $m$\superscript{th} through $n$\superscript{th} positions (inclusive) of $x$.
For $x \in \Sigma_2^*$, $\overline{x}$ denotes the binary complement of $x$, \ie, the word obtained by changing all {\tt 0}s to {\tt 1}s and vice versa.
We use the same notation just described for infinite words.
In addition, for ${\bf x} \in \Sigma^\omega$ and $n \in \bbN$, ${\bf x}[n\dotdot\infty]$ denotes the (infinite) suffix of ${\bf x}$ starting from the $n$\superscript{th} position of ${\bf x}$.

For a morphism $g : \Sigma^* \to \Sigma^*$ and $n \in \bbN$, we let $g^n$ denote the $n$-fold composition of $g$, and $g^\omega : \Sigma^* \to \Sigma^\omega$ denote $\lim_{n\to\infty} g^n$ if the limit exists.
The Thue-Morse morphism $\mu : \Sigma_2^* \to \Sigma_2^*$ is defined by $\mu({\tt 0}) \coloneq {\tt 01}$ and $\mu({\tt 1}) \coloneq {\tt 10}$.
Iterates of the Thue-Morse morphism acting on ${\tt 0}$ are denoted by $t_n \coloneq \mu^n({\tt 0})$.
Note that ${\bf t} = \mu^\omega({\tt 0})$.

\section{Similarity density of words}

Let us express Problem~\ref{prob:1} in another way: 
how similar can an arbitrary overlap-free word $\bf w$ be to ${\bf t}$?
For ${\bf w}$ a shift of ${\bf t}$, this was essentially determined by the following result from a surprisingly little-known 1927 paper of Kurt Mahler on autocorrelation \cite{Mahler:1927}.

\begin{theorem} \label{thm:mahler-autocorrelation}
For all $k \in \bbN$, the limit
\[
\sigma(k) \coloneq \lim_{n \rightarrow \infty} \frac{1}{n} \sum_{i=0}^{n-1} (-1)^{{\bf t}[i]+{\bf t}[i+k]}
\]
exists.
Furthermore, we have
$\sigma(0) = 1$, $\sigma(1) = -\frac{1}{3}$, and for all $n \in \bbN$, $\sigma(2n) = \sigma(n)$ and $\sigma(2n+1) = -\frac{1}{2}(\sigma(n) + \sigma(n+1))$.
\end{theorem}
(Also see \cite{Yarlagadda&Hershey:1984,Yarlagadda&Hershey:1990}.)
Then an easy induction on $k$ gives

\begin{corollary} \label{cor:mahler-autocorrelation}
For all $k \in \bbN \delete \{0\}$, $-\frac{1}{3} \leq \sigma(k) \leq \frac{1}{3}$.
\end{corollary}

Mahler's result is not exactly what we want, but we can easily transform it. Rather than autocorrelation, we are more interested in a quantity we call ``similarity density''; it measures how similar two words of the same length are, with a simple and intuitive definition for finite words that generalizes to infinite words by way of limits.

\begin{defn}
We interpret the Kronecker delta as a function of two variables $\delta : \Sigma^2 \to \Sigma_2$ as follows.
\[
\delta(a,b) \coloneq
\begin{cases}
0, &\text{if } a \neq b; \\
1, &\text{if } a = b.
\end{cases}
\]
\end{defn}

\begin{defn} \label{defn:sd}
Let $n \in \bbN \delete \{0\}$ and $x,y \in \Sigma^n$. The 
{\it similarity density} of $x$ and $y$ is
\[
\SD(x,y) \coloneq \frac{1}{n}\sum_{i=0}^{n-1} \delta(x[i],y[i]).
\]
\end{defn}

Thus, two finite words of the same length have similarity density $1$ if and only if they are equal.

\begin{defn} \label{defn:lusd}
Let ${\bf x}, {\bf y} \in \Sigma^\omega$. The {\it lower} and 
{\it upper similarity densities} of ${\bf x}$ and ${\bf y}$ are, respectively,
\begin{align*}
\LSD({\bf x}, {\bf y}) &\coloneq \liminf_{n\to\infty} \ \SD({\bf x}[0\dotdot n-1],{\bf y}[0\dotdot n-1]), \\
\USD({\bf x}, {\bf y}) &\coloneq \limsup_{n\to\infty} \ \SD({\bf x}[0\dotdot n-1],{\bf y}[0\dotdot n-1]).
\end{align*}
\end{defn}

\begin{remark} \label{rem:density-generalization}
Our notion of similarity density is not a new idea.
(Similar ideas can be found, e.g., in
\cite{Ochem&Rampersad&Shallit:2008,Grant&Shallit&Stoll:2009}.)
It is inspired by the well-studied number-theoretic notion of {\it asymptotic} or {\it natural density} of subsets of natural numbers.
The {\it lower} and {\it upper asymptotic densities} of $A \subseteq \bbN$ are, respectively,
\begin{align*}
\underline{d}(A) &\coloneq \liminf_{n\to\infty} \frac{1}{n}\abs{A \intersect \{0,\dots,n-1\}}, \\
\overline{d}(A) &\coloneq \limsup_{n\to\infty} \frac{1}{n}\abs{A \intersect \{0,\dots,n-1\}}.
\end{align*}
Similarity density generalizes asymptotic density in the following way.
For $A \subseteq \bbN$, let $\chi_A \in \Sigma_2^\omega$ denote the characteristic sequence of $A$ (\ie, $\chi_A[n] = {\tt 1}$ iff $n \in A$).
Then
\begin{align*}
\underline{d}(A) &= \LSD(\chi_A, {\tt 1}^\omega), \\
\overline{d}(A) &= \USD(\chi_A, {\tt 1}^\omega).
\end{align*}
\end{remark}

Mahler's result can now be restated as follows.

\begin{theorem} \label{thm:mahler-lusd}
For all $k \in \bbN \delete \{0\}$, $\frac{1}{3} \leq \LSD({\bf t}, {\bf t}[k\dotdot\infty]) = \USD({\bf t}, {\bf t}[k\dotdot\infty]) \leq \frac{2}{3}$.
\end{theorem}

\begin{proof}
Note that for all $i,k \in \bbN$, $(-1)^{{\bf t}[i]+{\bf t}[i+k]} = 2\delta({\bf t}[i], {\bf t}[i+k]) - 1$.
Hence, by Definition~\ref{defn:lusd}, Theorem~\ref{thm:mahler-autocorrelation}, and Corollary~\ref{cor:mahler-autocorrelation}, we obtain
\[
\LSD({\bf t}, {\bf t}[k\dotdot\infty]) = \USD({\bf t}, {\bf t}[k\dotdot\infty]) = \frac{1}{2}(\sigma(k)+1) \in \frac{1}{2}\left(\left[-\frac{1}{3},\frac{1}{3}\right]+1\right) = \left[\frac{1}{3},\frac{2}{3}\right]. \qedhere
\]
\end{proof}

\begin{remark} \label{rem:h}
There exist overlap-free infinite binary words ${\bf w}$ with $\LSD({\bf t}, {\bf w}) < \USD({\bf t}, {\bf w})$.
One example is the word ${\bf h} = {\tt 00100110100101100110100110010110 \cdots}$ whose $n$\superscript{th} position is the number of {\tt 0}s (modulo 2) in the binary representation of $n$.
(Note that ${\bf h}[0] = {\tt 0}$ as we take the binary representation of $0$ to be $\epsilon$.)
We prove in Proposition~\ref{prop:h} that $\LSD({\bf t}, {\bf h}) = \frac{1}{3}$ while $\USD({\bf t}, {\bf h}) = \frac{2}{3}$. See Figure~\ref{fig:SD-t-h}, where this similarity density is graphed as a function of the length of the prefix.

\begin{figure}[H]
\begin{center}
\includegraphics[width=5.0in]{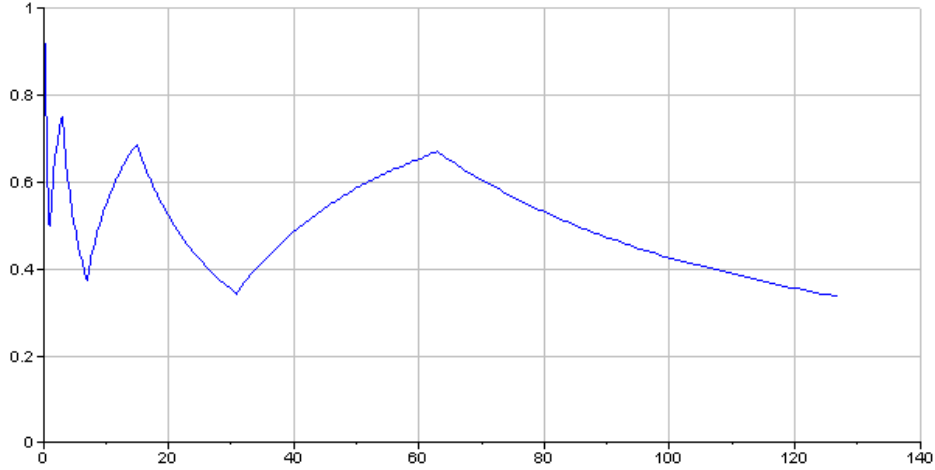}
\end{center}
%\vskip -.2in
\caption{Similarity density of prefixes of $\bf t$ and $\bf h$}
\protect\label{fig:SD-t-h}
\end{figure}
\end{remark}

Our main result (Theorem~\ref{thm:main}) is that the lower and upper similarity densities of $\bf t$ with {\it any} overlap-free infinite binary word other than ${\bf t}$ and $\overline{\bf t}$ are bounded below and above as in Theorem~\ref{thm:mahler-lusd}, but with the constants $\frac{1}{4}$ and $\frac{3}{4}$ instead of $\frac{1}{3}$ and $\frac{2}{3}$ respectively. However, computational evidence suggests that the tighest bounds are indeed $\frac{1}{3}$ and $\frac{2}{3}$, which, if true, would fully generalize Theorem~\ref{thm:mahler-lusd} from nontrivial shifts of ${\bf t}$ to all overlap-free infinite binary words (other than ${\bf t}$ and $\overline{\bf t}$).

The following are basic properties of similarity density that we will use later.
Their statements are all intuitive and their proofs are just basic exercises in algebra. Observation~\ref{obsv:sd-weighted-avg} states that similarity density can be computed using weighted averages.
Observation~\ref{obsv:complement-sd} and Corollary~\ref{cor:complement-lusd} explain how complementation affects similarity density.
Observation~\ref{obsv:lusd-ignore-finite} states that the similarity densities of infinite words depends only on their tails, so we can ignore arbitrarily long prefixes.
Observation~\ref{obsv:lusd-fixed-finite-interval} states that the similarity densities of infinite words can be obtained by considering similarity densities of prefixes where the length of the prefix grows by any constant instead of just by one in each iteration.

\begin{obsv} \label{obsv:sd-weighted-avg}
Let $n,m \in \bbN \delete \{0\}$, $u,v \in \Sigma^n$, and $x,y \in \Sigma^m$.
Then
\[
\SD(ux,vy) = \frac{n}{n+m}\SD(u,v) + \frac{m}{n+m}\SD(x,y).
\]
\end{obsv}

\begin{proof}
\begin{align*}
\SD(ux,vy)
&= \frac{1}{n+m}\sum_{i=0}^{n+m-1} \delta((ux)[i],(vy)[i]) \\
&= \frac{1}{n+m}\left(\sum_{i=0}^{n-1} \delta(u[i],v[i]) + \sum_{i=0}^{m-1} \delta(x[i],y[i])\right) \\
&= \frac{n}{n+m}\cdot \frac{1}{n}\sum_{i=0}^{n-1} \delta(u[i],v[i]) + \frac{m}{n+m}\cdot\frac{1}{m}\sum_{i=0}^{m-1} \delta(u[i],v[i]) \\
&= \frac{n}{n+m}\SD(u,v) + \frac{m}{n+m}\SD(x,y). \qedhere
\end{align*}
\end{proof}

\begin{obsv} \label{obsv:complement-sd}
For all $n \in \bbN \delete \{0\}$ and $x, y \in \Sigma_2^n$,
\begin{itemize}
\item[(i)] $\SD(\overline{x},y) = 1 - \SD(x,y)$.
\item[(ii)] $\SD(\overline{x},\overline{y}) = \SD(x,y)$.
\end{itemize}
\end{obsv}

\begin{proof}
\mbox{}
\begin{itemize}
\item[(i)] $\SD(\overline{x},y) = \frac{1}{n}\sum_{i=0}^{n-1} \delta(\overline{x}[i],y[i]) = \frac{1}{n}\sum_{i=0}^{n-1} (1 - \delta(x[i],y[i])) = 1 - \SD(x,y)$.
\item[(ii)] By (i) and symmetry of $\SD$, we have $\SD(\overline{x},\overline{y}) = 1 - \SD(x,\overline{y}) = 1 - (1 - \SD(x,y)) = \SD(x,y)$. \qedhere
\end{itemize}
\end{proof}

\begin{cor} \label{cor:complement-lusd}
For all ${\bf x}, {\bf y} \in \Sigma_2^\omega$,
\begin{itemize}
\item[(i)] $\LSD({\bf \overline{x}},{\bf y}) = 1 - \USD({\bf x},{\bf y})$ and $\USD({\bf \overline{x}},{\bf y}) = 1 - \LSD({\bf x},{\bf y})$.
\item[(ii)] $\LSD({\bf \overline{x}},{\bf \overline{y}}) = \LSD({\bf x},{\bf y})$ and $\USD({\bf \overline{x}},{\bf \overline{y}}) = \USD({\bf x},{\bf y})$.
\end{itemize}
\end{cor}

\begin{proof}
Immediate by Definition~\ref{defn:lusd}, Observation~\ref{obsv:complement-sd}, and basic properties of limits.
\end{proof}

\begin{obsv} \label{obsv:lusd-ignore-finite}
Let $l \in \bbN$, $u,v \in \Sigma^l$ and ${\bf x}, {\bf y} \in \Sigma^\omega$.
Then $\LSD(u{\bf x},v{\bf y}) = \LSD({\bf x}, {\bf y})$ and $\USD(u{\bf x},v{\bf y}) = \USD({\bf x}, {\bf y})$.
\end{obsv}

\begin{proof}
If $l = 0$, then the proof is trivial.
If $l > 0$, then we have
\begin{align*}
\LSD(u{\bf x},v{\bf y})
&= \liminf_{n\to\infty} \frac{1}{n}\sum_{i=0}^{n-1} \delta((u{\bf x})[i],(v{\bf y})[i]) \\
&= \liminf_{n\to\infty} \frac{1}{n+l}\sum_{i=0}^{n+l-1} \delta((u{\bf x})[i],(v{\bf y})[i]) \\
&= \liminf_{n\to\infty}
\left(\vphantom{\frac{1}{n+l}\sum_{i=0}^{l-1} \delta(u[i],v[i]) + \frac{1}{n+l}\sum_{i=0}^{n-1} \delta({\bf x}[i],{\bf y}[i])} \right.
\underbrace{\frac{1}{n+l}\sum_{i=0}^{l-1} \delta(u[i],v[i])}_{\in [0,\frac{l}{n+l}] \xrightarrow{n\to\infty} 0} + \frac{1}{n+l}\sum_{i=0}^{n-1} \delta({\bf x}[i],{\bf y}[i])
\left.\vphantom{\frac{1}{n+l}\sum_{i=0}^{l-1} (1 - \delta(u[i],v[i])) + \frac{1}{n+l}\sum_{i=0}^{n-1} (1 - \delta({\bf x}[i],{\bf y}[i]))} \right) \\
&= \liminf_{n\to\infty} \left(0 + \left(\frac{1}{n}-\frac{l}{n(n+l)}\right)\sum_{i=0}^{n-1} \delta({\bf x}[i],{\bf y}[i]) \right) \\
&= \liminf_{n\to\infty}
\left(\vphantom{\frac{1}{n}\sum_{i=0}^{n-1} \delta({\bf x}[i],{\bf y}[i]) - \frac{l}{n(n+l)}\sum_{i=0}^{n-1} \delta({\bf x}[i],{\bf y}[i])}\right.
\frac{1}{n}\sum_{i=0}^{n-1} \delta({\bf x}[i],{\bf y}[i]) - \underbrace{\frac{l}{n(n+l)}\sum_{i=0}^{n-1} (1 - \delta({\bf x}[i],{\bf y}[i]))}_{\in [0,\frac{l}{n+l}] \xrightarrow{n\to\infty} 0}
\left.\vphantom{\frac{1}{n}\sum_{i=0}^{n-1} (1 - \delta({\bf x}[i],{\bf y}[i])) - \frac{l}{n(n+l)}\sum_{i=0}^{n-1} (1 - \delta({\bf x}[i],{\bf y}[i]))}\right) \\
&= \liminf_{n\to\infty} \left(\frac{1}{n}\sum_{i=0}^{n-1} (1 - \delta({\bf x}[i],{\bf y}[i])) - 0 \right) \\
&= \LSD({\bf x},{\bf y}).
\end{align*}
The proof is exactly the same for $\USD$ with $\liminf$ replaced by $\limsup$.
\end{proof}

\begin{obsv} \label{obsv:lusd-fixed-finite-interval}
Let $M \in \bbN \delete \{0\}$.
Then
\begin{align*}
\LSD({\bf x}, {\bf y}) &= \liminf_{n\to\infty} \ \SD({\bf x}[0\dotdot Mn-1],{\bf y}[0\dotdot Mn-1]), \\
\USD({\bf x}, {\bf y}) &= \limsup_{n\to\infty} \ \SD({\bf x}[0\dotdot Mn-1],{\bf y}[0\dotdot Mn-1]).
\end{align*}
\end{obsv}

\begin{proof}
For any $n \in \bbN \delete \{0\}$ and $k \in \{Mn,Mn+1,\dots,M(n+1)-2\}$, by Observation~\ref{obsv:sd-weighted-avg}, we have
\begin{align*}
\SD({\bf x}[0\dotdot k],{\bf y}[0\dotdot k])
&= \frac{Mn}{k+1} \SD({\bf x}[0\dotdot Mn-1],{\bf y}[0\dotdot Mn-1]) \\
&\quad {} + \frac{k-Mn+1}{k+1} \SD({\bf x}[Mn\dotdot k],{\bf y}[Mn\dotdot k]) \\
&\in \left[\frac{Mn}{M(n+1)-1},\frac{Mn}{Mn+1}\right] \SD({\bf x}[0\dotdot Mn-1],{\bf y}[0\dotdot Mn-1]) \\
&\quad {} + \left[\frac{1}{M(n+1)-1},\frac{M-1}{Mn+1}\right] \SD({\bf x}[Mn\dotdot k],{\bf y}[Mn\dotdot k]),
\end{align*}
so since $\lim_{n\to\infty} [\frac{Mn}{M(n+1)-1},\frac{Mn}{Mn+1}] = [1,1] = \{1\}$ and $\lim_{n\to\infty} [\frac{1}{M(n+1)-1},\frac{M-1}{Mn+1}] = [0,0] = \{0\}$, all of the intermediate values $\SD({\bf x}[0\dotdot k],{\bf y}[0\dotdot k])$ for $k \in \{Mn,Mn+1,\dotdot,M(n+1)-2\}$ get arbitrarily close to $\SD({\bf x}[0\dotdot Mn-1],{\bf y}[0\dotdot Mn-1])$ as $n \to \infty$.
Hence,
\begin{align*}
\liminf_{n\to\infty}\ \SD({\bf x}[0\dotdot n-1],{\bf y}[0\dotdot n-1])
&= \liminf_{n\to\infty} \ \SD({\bf x}[0\dotdot Mn-1],{\bf y}[0\dotdot Mn-1]), \\
\limsup_{n\to\infty}\ \SD({\bf x}[0\dotdot n-1],{\bf y}[0\dotdot n-1])
&= \limsup_{n\to\infty}\ \SD({\bf x}[0\dotdot Mn-1],{\bf y}[0\dotdot Mn-1]). \qedhere
\end{align*}
\end{proof}

\section{Fife automaton for overlap-free infinite binary words}

We recall the so-called ``Fife automaton'' for overlap-free infinite binary words from \cite{Shallit:2011}.  (Note that this automaton does not appear
in the original paper of Fife \cite{Fife:1980}.)

\begin{figure}[H]
\begin{center}
\includegraphics[scale=0.75]{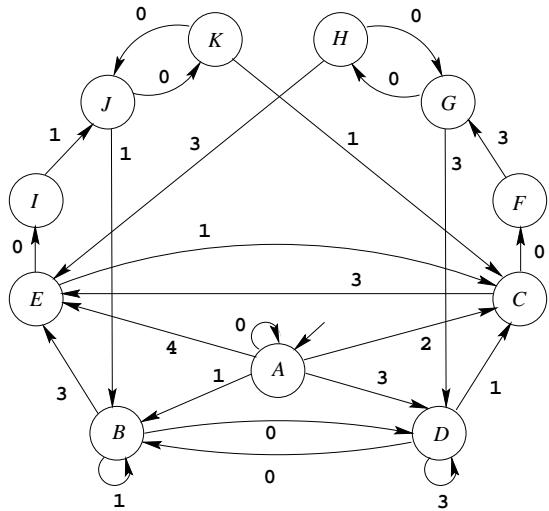}
\end{center}
\caption{Automaton encoding all overlap-free infinite binary words}
\label{nice}
\end{figure}

Here, infinite paths through the automaton encode all overlap-free infinite binary words, as follows.

\begin{defn} \label{defn:fife-bin}
First, each of the edge labels encodes a binary word, via $c : \Sigma_5 \to \Sigma_2^*$ defined by
\begin{align*}
c({\tt 0}) &\coloneq \epsilon, \\
c({\tt 1}) &\coloneq {\tt 0}, \\
c({\tt 2}) &\coloneq {\tt 00}, \\
c({\tt 3}) &\coloneq {\tt 1}, \\
c({\tt 4}) &\coloneq {\tt 11} .
\end{align*}

Then, the Fife-to-binary encoding $\FBE : (\Sigma_5^\omega \delete \Sigma_5^*{\tt 0}^\omega) \union \left(\Sigma_5^*{\tt 0}^\omega \times \Sigma_2\right) \to \Sigma_2^\omega$ is defined by
\begin{align*}
\FBE({\bf x}) &\coloneq \prod_{n=0}^\infty \mu^n(c({\bf x}[n])) &&\text{for } {\bf x} \in \Sigma_5^\omega \delete \Sigma_5^*{\tt 0}^\omega; \\
\FBE({\bf x}, a) &\coloneq \left(\prod_{n=0}^\infty \mu^n(c({\bf x}[n]))\right) \mu^\omega(a) &&\text{for } ({\bf x},a) \in \Sigma_5^*{\tt 0}^\omega \times \Sigma_2.
\end{align*}
Note that $\FBE$ is well-defined because $c$ is only erasing for the letter ${\tt 0}$ and $\mu$ is non-erasing,
so for ${\bf x} \in \Sigma_5^\omega$, the concatenation
$\prod_{n=0}^\infty \mu^n(c({\bf x}[n]))$ is finite iff ${\bf x}$ ends in ${\tt 0}^\omega$.
\end{defn}

We now recall the basic property of the automaton from \cite{Shallit:2011}.

\begin{thm} \label{thm:ovlf-words}
Let ${\bf w} \in \Sigma_2^\omega$. Then ${\bf w}$ is overlap-free iff there exists ${\bf x} \in \Sigma_5^\omega$ that encodes a valid path through the Fife automaton for overlap-free infinite binary words such that $\FBE({\bf x}) = {\bf w}$ (if ${\bf x}$ does not end in ${\tt 0}^\omega$) or $\FBE({\bf x},a) = {\bf w}$ (if ${\bf x}$ ends in ${\tt 0}^\omega$) for some $a \in S$, where $S \subseteq \Sigma_2$ depends on the eventual cycle corresponding to the suffix ${\tt 0}^\omega$ of the path encoded by ${\bf x}$: on state $A$ and between states $B$ and $D$ ($S = \Sigma_2$), between states $G$ and $H$ ($S = \{{\tt 1}\})$, or between states $J$ and $K$ ($S = \{{\tt 0}\}$).
\end{thm}

Recall ${\bf h}$ as defined in Remark~\ref{rem:h}.
Note that the definitions of ${\bf h}$ and ${\bf t}$ are very similar.
This is related to the special path that encodes ${\bf h}$ in the Fife automaton for overlap-free infinite binary words \cite{Shallit:2011}: ${\bf h} = \FBE({\tt 2}({\tt 31})^\omega)$.
We will see later in our proof of our main result why this path is special.
For now, we can use this path to compute the following result.

\begin{prop} \label{prop:h}
$\LSD({\bf h},{\bf t}) = \LSD({\bf \overline{h}},{\bf t}) = \frac{1}{3}$ and $\USD({\bf h},{\bf t}) = \USD({\bf \overline{h}},{\bf t}) = \frac{2}{3}$.
\end{prop}

\begin{proof}
Note that
\begin{align*}
{\bf h} &= \FBE({\tt 2}({\tt 31})^\omega)
= \mu^0(p({\tt 2})) \prod_{n=0}^\infty \left(\mu^{2n+1}(p({\tt 3})) \mu^{2n+2}(p({\tt 1}))\right) \\
&= \mu^0({\tt 00}) \prod_{n=0}^\infty \left(\mu^{2n+1}({\tt 1}) \mu^{2n+2}({\tt 0})\right) 
= {\tt 0}t_0 \prod_{n=0}^\infty \left(\overline{t_{2n+1}} t_{2n+2}\right) 
= {\tt 0}\prod_{n=0}^\infty \left(t_{2n} \overline{t_{2n+1}}\right),
\end{align*}
and since for each $n \in \bbN$, we have $\abs{t_n} = 2^n$ and $1 + \sum_{i=0}^n 2^i = 2^{n+1}$, it follows that
\[
{\bf h}[2^n \dotdot 2^{n+1} - 1] =
\begin{cases}
t_n, &\text{if } n \equiv 0 \pmod 2; \\
\overline{t_n}, &\text{if } n \equiv 1 \pmod 2. \\
\end{cases}
\]
Note that for each $n \in \bbN$, we have ${\bf t}[2^n\dotdot 2^{n+1}-1] = {t_{n+1}}[2^n\dotdot 2^{n+1}-1] = \overline{t_n}$.
Hence, for all $n \in \bbN$,
\[
{\bf h}[2^n\dotdot 2^{n+1}-1] = 
\begin{cases}
{\bf \overline{t}}[2^n\dotdot 2^{n+1}-1], &\text{if } n \equiv 0 \pmod 2; \\
{\bf t}[2^n\dotdot 2^{n+1}-1], &\text{if } n \equiv 1 \pmod 2, \\
\end{cases}
\]
whence
\[
\SD({\bf h}[2^n\dotdot 2^{n+1}-1],{\bf t}[2^n\dotdot 2^{n+1}-1]) =
\begin{cases}
0, &\text{if } n \equiv 0 \pmod 2; \\
1,  &\text{if } n \equiv 1 \pmod 2. \\
\end{cases}
\]
If we consider two of these blocks at a time, we obtain, by Observation~\ref{obsv:sd-weighted-avg}, that for all $n \in \bbN$,
\begin{align*}
\SD({\bf h}[2^n\dotdot 2^{n+2}-1],{\bf t}[2^n\dotdot 2^{n+2}-1])
&= \frac{2^n}{2^n+2^{n+1}}\SD({\bf h}[2^n\dotdot 2^{n+1}-1],{\bf t}[2^n\dotdot 2^{n+1}-1]) \\
&\quad {} + \frac{2^{n+1}}{2^n+2^{n+1}}\SD({\bf h}[2^{n+1}\dotdot 2^{n+2}-1],{\bf t}[2^{n+1}\dotdot 2^{n+2}-1]) \\
&=
\begin{cases}
\frac{2}{3}, &\text{if } n \equiv 0 \pmod 2; \\
\frac{1}{3}, &\text{if } n \equiv 1 \pmod 2 .\\
\end{cases}
\end{align*}
Iterating Observation~\ref{obsv:sd-weighted-avg} finitely many times, we obtain that for all $n \in \bbN$,
\begin{gather*}
\SD({\bf h}[1\dotdot 2^{2n}-1],{\bf t}[1\dotdot 2^{2n}-1]) = \frac{2}{3}, \\
\SD({\bf h}[2\dotdot 2^{2n+1}-1],{\bf t}[2\dotdot 2^{2n+1}-1]) = \frac{1}{3}.
\end{gather*}
Furthermore, applying Observation~\ref{obsv:sd-weighted-avg} one letter at a time, we see that for $k \in [2^{2n}-1,2^{2n+1}-1]$, $\SD({\bf h}[1\dotdot k],{\bf t}[1\dotdot k])$ monotonically decreases (from $\frac{2}{3}$), and for $k \in [2^{2n+1}-1,2^{2n+2}-1]$, $\SD({\bf h}[1\dotdot k],{\bf t}[1\dotdot k])$ monotonically increases (back to $\frac{2}{3}$).
Thus,
\[
\USD({\bf h}[1\dotdot\infty],{\bf t}[1\dotdot\infty]) = \limsup_{n\to\infty}\ \SD({\bf h}[1\dotdot n],{\bf t}[1\dotdot n]) = \frac{2}{3}.
\]
Similarly, for $k \in [2^{2n+1}-1,2^{2n+2}-1]$, $\SD({\bf h}[2\dotdot k],{\bf t}[2\dotdot k])$ monotonically increases (from $\frac{1}{3}$), and for $k \in [2^{2n+2}-1,2^{2n+3}-1]$, $\SD({\bf h}[2\dotdot k],{\bf t}[2\dotdot k])$ monotonically decreases (back to $\frac{1}{3}$), so
\[
\LSD({\bf h}[2\dotdot\infty],{\bf t}[2\dotdot\infty]) = \liminf_{n\to\infty} \ \SD({\bf h}[2\dotdot n+1],{\bf t}[2\dotdot n+1]) = \frac{1}{3}.
\]
Finally, by Observation~\ref{obsv:lusd-ignore-finite}, we conclude that $\LSD({\bf h},{\bf t}) = \LSD({\bf h}[2\dotdot\infty],{\bf t}[2\dotdot\infty]) = \frac{1}{3}$ and $\USD({\bf h},{\bf t}) = \USD({\bf h}[1\dotdot\infty],{\bf t}[1\dotdot\infty]) = \frac{2}{3}$, whence by Corollary~\ref{cor:complement-lusd}(i), we obtain $\LSD({\bf \overline{h}},{\bf t}) = 1 - \USD({\bf h},{\bf t}) = 1 - \frac{2}{3} = \frac{1}{3}$ and $\USD({\bf \overline{h}},{\bf t}) = 1 - \LSD({\bf h},{\bf t}) = 1 - \frac{1}{3} = \frac{2}{3}$.
\end{proof}

\section{Main result}

We now state and prove our main result.

\begin{thm} \label{thm:main}
For all overlap-free ${\bf w} \in \Sigma_2^\omega \delete \{{\bf t}, {\bf \overline{t}}\}$, $\frac{1}{4} \leq \LSD({\bf w},{\bf t}) \leq \USD({\bf w},{\bf t}) \leq \frac{3}{4}$.
\end{thm}

Our approach to proving Theorem~\ref{thm:main} is to consider each overlap-free infinite binary word in terms of the path through the Fife automaton that encodes it.
We divide the paths into four cases.
\begin{enumerate}[label=(\arabic*)]
\item ends in ${\tt 0}^\omega$.
\item does not end in ${\tt 0}^\omega$, begins with ${\tt 0}^n{\tt 2}$ or ${\tt 0}^n{\tt 4}$ for some $n \in \bbN$, and contains exactly $n$ ${\tt 0}$s.
\item does not end in ${\tt 0}^\omega$, begins with ${\tt 0}^n{\tt 2}$ or ${\tt 0}^n{\tt 4}$ for some $n \in \bbN$, and contains more than $n$ ${\tt 0}$s.
\item does not end in ${\tt 0}^\omega$ and begins with ${\tt 0}^n{\tt 1}$ or ${\tt 0}^n{\tt 3}$ for some $n \in \bbN$.
\end{enumerate}
Upon closer examination of the Fife automaton, case (2) can be subdivided into two cases: ${\tt 0}^n{\tt 2(31)}^\omega$ and their complements under $\FBE$, ${\tt 0}^n{\tt 4(13)}^\omega$.
It turns out that we can bootstrap Proposition~\ref{prop:h} to obtain the same bounds for both of these cases.
Case (1) follows from Mahler's theorem~\ref{thm:mahler-lusd}, but it will also follow from our own generalized version of it (albeit with weaker bounds).
For cases (3) and (4), we observe that the infinite binary word corresponding to the path eventually ``lags behind'' the prefixes $t_n$ of ${\bf t}$ in the sense that each successive $n$\superscript{th} symbol in the path can only generate positions prior to $2^n$, whence we can use a technical lemma that bounds the similarity density of $t_n$ with nontrivial subwords of $t_{n+1}$ to complete the proof.

\begin{prop} \label{prop:edge}
For all $n \in \bbN$ we have
$\LSD(\FBE({\tt 0}^n{\tt 2(31)}^\omega),{\bf t}) = \frac{1}{3}$ and $\USD(\FBE({\tt 0}^n{\tt 2(31)}^\omega),{\bf t}) = \frac{2}{3}$.
\end{prop}

\begin{proof}
Note that
\begin{align*}
\FBE({\tt 0}^n{\tt 2}({\tt 31})^\omega)
&= \prod_{k=0}^{n-1} \left(\mu^k(p({\tt 0}))\right)\mu^n(p({\tt 2})) \prod_{k=0}^\infty \left(\mu^{n+2k+1}(p({\tt 3})) \mu^{n+2k+2}(p({\tt 1}))\right) \\
&= \prod_{k=0}^{n-1} \left(\mu^k(\epsilon)\right)\mu^n({\tt 00}) \prod_{k=0}^\infty \left(\mu^{n+2k+1}({\tt 1}) \mu^{n+2k+2}({\tt 0})\right) \\
&= t_n t_n \prod_{k=0}^\infty \left(\overline{t_{n+2k+1}} t_{n+2k+2}\right) \\
&= t_n \prod_{k=0}^\infty \left(t_{n+2k}\overline{t_{n+2k+1}}\right).
\end{align*}
From the proof of Proposition~\ref{prop:h}, we see that
\begin{align*}
\FBE({\tt 0}^n{\tt 2}({\tt 31})^\omega)[2^n\dotdot\infty] =
\begin{cases}
{\bf h}[2^n\dotdot\infty], &\text{if } n \equiv 0 \pmod 2; \\
{\bf \overline{h}}[2^n\dotdot\infty], &\text{if } n \equiv 1 \pmod 2 .
\end{cases}
\end{align*}
Hence, by Observation~\ref{obsv:lusd-ignore-finite} and Proposition~\ref{prop:h}, we have
\begin{align*}
(\LSD,\USD)(\FBE({\tt 0}^n{\tt 2}({\tt 31})^\omega),{\bf t})
&= (\LSD,\USD)(\FBE({\tt 0}^n{\tt 2}({\tt 31})^\omega)[2^n\dotdot\infty],{\bf t}[2^n\dotdot\infty]) \\
&= \begin{cases}
(\LSD,\USD)({\bf h}[2^n\dotdot\infty],{\bf t}[2^n\dotdot\infty]), &\text{if } n \equiv 0 \pmod 2; \\
(\LSD,\USD)({\bf \overline{h}}[2^n\dotdot\infty],{\bf t}[2^n\dotdot\infty]), &\text{if } n \equiv 1 \pmod 2,
\end{cases} \\
&= \begin{cases}
(\LSD,\USD)({\bf h},{\bf t}), &\text{if } n \equiv 0 \pmod 2; \\
(\LSD,\USD)({\bf \overline{h}},{\bf t}), &\text{if } n \equiv 1 \pmod 2,
\end{cases} \\
&= \left(\frac{1}{3},\frac{2}{3}\right). \qedhere
\end{align*}
\end{proof}

\begin{lem} \label{lem:technical-one}
For all $n \in \bbN$ and $i \in [1,2^n-1]$,
\begin{align*}
\text{(a)} &&
\SD(t_n, t_{n+1}[i \dotdot 2^n + i - 1]) &\in
\begin{cases}
\{\frac{1}{2}\},  &\text{if } i = 2^{n-1} ; \\
[\frac{1}{4},\frac{3}{4}], &\text{otherwise} .
\end{cases} \\
\text{$\overline{(}$a)} &&
\SD(\overline{t_n}, t_{n+1}[i \dotdot 2^n + i - 1]) &\in
\begin{cases}
\{\frac{1}{2}\}, &\text{if } i = 2^{n-1}; \\
[\frac{1}{4},\frac{3}{4}], &\text{otherwise}.
\end{cases} \\
\text{(a$\overline{)}$} &&
\SD(t_n, \overline{t_{n+1}}[i \dotdot 2^n + i - 1]) &\in
\begin{cases}
\{\frac{1}{2}\}, &\text{if } i = 2^{n-1} ; \\
[\frac{1}{4},\frac{3}{4}], &\text{otherwise}.
\end{cases} \\
\text{$\overline{(}$a$\overline{)}$} &&
\SD(\overline{t_n}, \overline{t_{n+1}}[i \dotdot 2^n + i - 1]) &\in
\begin{cases}
\{\frac{1}{2}\}, &\text{if } i = 2^{n-1} ; \\
[\frac{1}{4},\frac{3}{4}], &\text{otherwise}.
\end{cases} \\
\text{(b)} &&
\SD(t_n, t_n^2[i \dotdot 2^n + i - 1]) &\in
\begin{cases}
\{0\}, &\text{if } i = 2^{n-1}; \\
[\frac{1}{4},\frac{3}{4}], &\text{otherwise}.
\end{cases} \\
\text{$\overline{(}$b)} &&
\SD(\overline{t_n}, t_n^2[i \dotdot 2^n + i - 1]) &\in
\begin{cases}
\{1\}, &\text{if } i = 2^{n-1}; \\
[\frac{1}{4},\frac{3}{4}], &\text{otherwise}.
\end{cases} \\
\text{(b$\overline{)}$} &&
\SD(t_n, \overline{t_n}^2[i \dotdot 2^n + i - 1]) &\in
\begin{cases}
\{1\}, &\text{if } i = 2^{n-1}; \\
[\frac{1}{4},\frac{3}{4}], &\text{otherwise}.
\end{cases} \\
\text{$\overline{(}$b$\overline{)}$} &&
\SD(\overline{t_n}, \overline{t_n}^2[i \dotdot 2^n + i - 1]) &\in
\begin{cases}
\{0\}, &\text{if } i = 2^{n-1}; \\
[\frac{1}{4},\frac{3}{4}], &\text{otherwise} .
\end{cases}
\end{align*}
\end{lem}

\begin{proof}
By induction on $n$.
\begin{itemize}
\item For $n = 0$, all eight cases are vacuously true due to $i \in \emptyset$.
\item Suppose all eight cases hold for some $n \in \bbN$.
For $i \in [1,2^{n+1}-1]$, using Observation~\ref{obsv:sd-weighted-avg} followed by the induction hypothesis, we calculate
\begin{align*}
&\SD(t_{n+1}, t_{n+2}[i \dotdot 2^{n+1} + i - 1]) \\
&= \SD(t_n\overline{t_n}, (t_n\overline{t_n}\overline{t_n}t_n)[i \dotdot 2^{n+1} + i - 1]) \\
&=
\begin{cases}
\SD(t_n\overline{t_n}, (t_n\overline{t_n}\overline{t_n})[i \dotdot 2^{n+1} + i - 1]), &\text{if } i \in [1,2^n-1]; \\
\SD(t_n\overline{t_n}, \overline{t_n}\overline{t_n}), &\text{if } i = 2^n; \\
\SD(t_n\overline{t_n}, (\overline{t_n}\overline{t_n}t_n)[i - 2^n \dotdot 2^n + i - 1]), &\text{if } i \in [2^n+1,2^{n+1}-1], \\
\end{cases} \\
&=
\begin{cases}
\frac{2^n}{2^{n+1}}\SD(t_n, (t_n\overline{t_n})[i \dotdot 2^n + i - 1]) + \frac{2^n}{2^{n+1}}\SD(\overline{t_n}, (\overline{t_n}\overline{t_n})[i \dotdot 2^n + i - 1]), &\text{if } i \in [1,2^n-1]; \\
\frac{2^n}{2^{n+1}}\SD(t_n, \overline{t_n}) + \frac{2^n}{2^{n+1}}\SD(\overline{t_n}, \overline{t_n}), &\text{if } i = 2^n; \\
\frac{2^n}{2^{n+1}}\SD(t_n, (\overline{t_n}\overline{t_n})[i - 2^n \dotdot i - 1]) + \frac{2^n}{2^{n+1}}\SD(\overline{t_n}, (\overline{t_n}t_n)[i - 2^n \dotdot i - 1]), &\text{if } i \in [2^n+1,2^{n+1}-1],
\end{cases} \\
&\in
\begin{cases}
\frac{1}{2}\{\frac{1}{2}\} + \frac{1}{2}\{0\}, &\text{if } i = 2^{n-1}; \text{ }(\text{by }(\text{a}), \overline{(}\text{b}\overline{)}) \\
\frac{1}{2}[\frac{1}{4},\frac{3}{4}] + \frac{1}{2}[\frac{1}{4},\frac{3}{4}], &\text{if } i \in [1,2^n-1] \delete \{2^{n-1}\}; \text{ }(\text{by }(\text{a}), \overline{(}\text{b}\overline{)}) \\
\frac{1}{2}\{0\} + \frac{1}{2}\{1\}, &\text{if } i = 2^n; \\
\frac{1}{2}\{1\} + \frac{1}{2}\{\frac{1}{2}\}, &\text{if } i = 2^n + 2^{n-1}; \text{ }(\text{by }(\text{b}\overline{)}, \overline{(}\text{a}\overline{)}) \\
\frac{1}{2}[\frac{1}{4},\frac{3}{4}] + \frac{1}{2}[\frac{1}{4},\frac{3}{4}], &\text{if } i \in [2^n+1,2^{n+1}-1] \delete \{2^n + 2^{n-1}\}, \text{ }(\text{by }(\text{b}\overline{)}, \overline{(}\text{a}\overline{)})
\end{cases} \\
&=
\begin{cases}
\{\frac{1}{4}\}, &\text{if } i = 2^{n-1}; \\
[\frac{1}{4},\frac{3}{4}], &\text{if } i \in [1,2^n-1] \delete \{2^{n-1}\}; \\
\{\frac{1}{2}\}, &\text{if } i = 2^n; \\
\{\frac{3}{4}\}, &\text{if } i = 2^n + 2^{n-1}; \\
[\frac{1}{4},\frac{3}{4}], &\text{if } i \in [2^n+1,2^{n+1}-1] \delete \{2^n + 2^{n-1}\},
\end{cases}
\end{align*}
\begin{align*}
&\SD(t_{n+1}, t_{n+1}^2[i \dotdot 2^{n+1} + i - 1]) \\
&= \SD(t_n\overline{t_n}, (t_n\overline{t_n}t_n\overline{t_n})[i \dotdot 2^{n+1} + i - 1]) \\
&=
\begin{cases}
\SD(t_n\overline{t_n}, (t_n\overline{t_n}t_n)[i \dotdot 2^{n+1} + i - 1]), &\text{if } i \in [1,2^n-1]; \\
\SD(t_n\overline{t_n}, \overline{t_n}t_n), &\text{if } i = 2^n; \\
\SD(t_n\overline{t_n}, (\overline{t_n}t_n\overline{t_n})[i - 2^n \dotdot 2^n + i - 1]), &\text{if } i \in [2^n+1,2^{n+1}-1], \\
\end{cases} \\
&=
\begin{cases}
\frac{2^n}{2^{n+1}}\SD(t_n, (t_n\overline{t_n})[i \dotdot 2^n + i - 1]) + \frac{2^n}{2^{n+1}}\SD(\overline{t_n}, (\overline{t_n}t_n)[i \dotdot 2^n + i - 1]), &\text{if } i \in [1,2^n-1]; \\
\frac{2^n}{2^{n+1}}\SD(t_n, \overline{t_n}) + \frac{2^n}{2^{n+1}}\SD(\overline{t_n}, t_n), &\text{if } i = 2^n; \\
\frac{2^n}{2^{n+1}}\SD(t_n, (\overline{t_n}t_n)[i - 2^n \dotdot i - 1]) + \frac{2^n}{2^{n+1}}\SD(\overline{t_n}, (t_n\overline{t_n})[i - 2^n \dotdot i - 1]), &\text{if } i \in [2^n+1,2^{n+1}-1],
\end{cases} \\
&\in
\begin{cases}
\frac{1}{2}\{\frac{1}{2}\} + \frac{1}{2}\{\frac{1}{2}\}, &\text{if } i = 2^{n-1}; \text{ }(\text{by }(\text{a}), \overline{(}\text{a}\overline{)}) \\
\frac{1}{2}[\frac{1}{4},\frac{3}{4}] + \frac{1}{2}[\frac{1}{4},\frac{3}{4}], &\text{if } i \in [1,2^n-1] \delete \{2^{n-1}\}; \text{ }(\text{by }(\text{a}), \overline{(}\text{a}\overline{)}) \\
\frac{1}{2}\{0\} + \frac{1}{2}\{0\}, &\text{if } i = 2^n; \\
\frac{1}{2}\{\frac{1}{2}\} + \frac{1}{2}\{\frac{1}{2}\}, &\text{if } i = 2^n + 2^{n-1}; \text{ }(\text{by }(\text{a}\overline{)}, \overline{(}\text{a})) \\
\frac{1}{2}[\frac{1}{4},\frac{3}{4}] + \frac{1}{2}[\frac{1}{4},\frac{3}{4}], &\text{if } i \in [2^n+1,2^{n+1}-1] \delete \{2^n + 2^{n-1}\}, \text{ }(\text{by }(\text{a}\overline{)}, \overline{(}\text{a}))
\end{cases} \\
&=
\begin{cases}
\{\frac{1}{2}\}, &\text{if } i = 2^{n-1}; \\
[\frac{1}{4},\frac{3}{4}], &\text{if } i \in [1,2^n-1] \delete \{2^{n-1}\}; \\
0, &\text{if } i = 2^n; \\
\{\frac{1}{2}\}, &\text{if } i = 2^n + 2^{n-1}; \\
[\frac{1}{4},\frac{3}{4}], &\text{if } i \in [2^n+1,2^{n+1}-1] \delete \{2^n + 2^{n-1}\},
\end{cases}
\end{align*}
hence proving (a) and (b) also hold for $n+1$. By Observation~\ref{obsv:complement-sd}, the remaining six cases also hold for $n+1$. \qedhere
\end{itemize}
\end{proof}

\begin{cor} \label{cor:technical-finite}
For all $n \in \bbN$, $i \in [0,2^n-1]$ with $\gcd(i,2^n) \leq 2^{n-2}$, and $x,y_0,y_1 \in \{t_n,\overline{t_n}\}$,
\[
\SD(x,(y_0y_1)[i \dotdot i + 2^n - 1]) \in [\tfrac{1}{4},\tfrac{3}{4}].
\]
\end{cor}

\begin{proof}
Follows immediately from Lemma~\ref{lem:technical-one}.
\end{proof}

\begin{cor} \label{cor:technical-infinite}
For all $n, i \in \bbN$ with $\gcd(i,2^n) \leq 2^{n-2}$ and ${\bf x}, {\bf y} \in \{t_n,\overline{t_n}\}^\omega$,
\[
\frac{1}{4} \leq \LSD({\bf x},{\bf y}[i \dotdot \infty]) \leq \USD({\bf x},{\bf y}[i \dotdot \infty]) \leq \frac{3}{4}.
\]
\end{cor}

\begin{proof}
Note that for any $j \in \bbN$, $\gcd(i + j \cdot 2^n, 2^n) = \gcd(i, 2^n) \leq 2^{n-2}$.
Also for any $j \in \bbN$, since ${\bf x},{\bf y} \in \{t_n,\overline{t_n}\}^\omega$ and $\abs{t_n} = \abs{\overline{t_n}} = 2^n$, we have ${\bf x}[2^nj \dotdot 2^n(j+1)-1] \in \{t_n,\overline{t_n}\}$ and ${\bf y}[i + 2^nj \dotdot i + 2^n(j+1) - 1] = (y_0y_1)[(i \bmod 2^n) \dotdot (i \bmod 2^n) + 2^n - 1]$ for some $y_0,y_1 \in \{t_n,\overline{t_n}\}$.
Hence, for any $j \in \bbN$, by Corollary~\ref{cor:technical-finite},
\[
\SD({\bf x}[2^nj \dotdot 2^n(j+1)-1], {\bf y}[i+2^nj \dotdot i+2^n(j+1)-1]) \in \left[\frac{1}{4},\frac{3}{4}\right],
\]
whence by Observation~\ref{obsv:sd-weighted-avg},
\[
\SD({\bf x}[0 \dotdot 2^n(j+1)-1], {\bf y}[i \dotdot i+2^n(j+1)-1]) \in \left[\frac{1}{4},\frac{3}{4}\right],
\]
whence by Observation~\ref{obsv:lusd-fixed-finite-interval},
\begin{align*}
(\LSD,\USD)({\bf x},{\bf y}[i\dotdot\infty])
&= \left(\liminf_{j\to\infty},\limsup_{j\to\infty}\right) \SD({\bf x}[0 \dotdot 2^nj-1], {\bf y}[i\dotdot i+2^nj-1]) \\
&\in \left(\left[\frac{1}{4},\frac{3}{4}\right],\left[\frac{1}{4},\frac{3}{4}\right]\right). \qedhere
\end{align*}
\end{proof}

\begin{cor} \label{cor:technical-infinite-t}
For all $i \in \bbN \delete \{0\}$, $\frac{1}{4} \leq \LSD({\bf t}, {\bf t}[i\dotdot\infty]) \leq \USD({\bf t}, {\bf t}[i\dotdot\infty]) \leq \frac{3}{4}$.
\end{cor}

\begin{proof}
Since $i > 0$, we have
$4\max_{m \in \bbN} \gcd(i,2^m) = 2^n$ for some $n \in \bbN$.
Note that $\gcd(i,2^n) = 2^{n-2}$.
Also note that ${\bf t} = \mu^n({\bf t}) \in \{t_n,\overline{t_n}\}^\omega$.
Hence, by Corollary~\ref{cor:technical-infinite},
\[
\frac{1}{4} \leq \LSD({\bf t},{\bf t}[i \dotdot \infty]) \leq \USD({\bf t},{\bf t}[i \dotdot \infty]) \leq \frac{3}{4}. \qedhere
\]
\end{proof}

We now have all the tools needed to prove Theorem~\ref{thm:main}.

\begin{proof}[Proof of Theorem~\ref{thm:main}]
Let ${\bf w} \in \Sigma_2^\omega \delete \{{\bf t},{\bf \overline{t}}\}$.
By Theorem~\ref{thm:ovlf-words}, there exists ${\bf x} \in \Sigma_5^\omega$ that encodes a valid path through the Fife automaton for overlap-free infinite binary words such that $\FBE({\bf x}) = {\bf w}$ or $\FBE({\bf x},a) = {\bf w}$ for some $a \in \Sigma_2$.
From inspection of the Fife automaton for overlap-free infinite binary words, we see that ${\bf x}$ must fall into one of the following four cases.
\begin{enumerate}[label=(\arabic*)]
\item ${\bf x}$ ends in ${\tt 0}^\omega$.
\item ${\bf x}$ does not end in ${\tt 0}^\omega$, begins with ${\tt 0}^n{\tt 2}$ or ${\tt 0}^n{\tt 4}$ for some $n \in \bbN$, and contains exactly $n$ ${\tt 0}$s.
\item ${\bf x}$ does not end in ${\tt 0}^\omega$, begins with ${\tt 0}^n{\tt 2}$ or ${\tt 0}^n{\tt 4}$ for some $n \in \bbN$, and contains more than $n$ ${\tt 0}$s.
\item ${\bf x}$ does not end in ${\tt 0}^\omega$ and begins with ${\tt 0}^n{\tt 1}$ or ${\tt 0}^n{\tt 3}$ for some $n \in \bbN$.
\end{enumerate}

\begin{enumerate}[label={\bf Case \arabic*:}]
{\setlength\itemindent{0.5cm}
\item
${\bf w}$ ends in either ${\bf t}$ or ${\bf \overline{t}}$, so since ${\bf w} \not\in \{{\bf t}, {\bf \overline{t}}\}$, it follows that ${\bf w} \in \{ z{\bf t}, z{\bf \overline{t}} \}$ for some $z \in \Sigma_2^+$.
By Observation~\ref{obsv:lusd-ignore-finite}, we have 
$$(\LSD,\USD)({\bf w},{\bf t}) \in \{ (\LSD,\USD)({\bf t},{\bf t}[\abs{z}\dotdot\infty]), (\LSD,\USD)({\bf \overline{t}},{\bf t}[\abs{z}\dotdot\infty]) \},$$
whence by Corollary~\ref{cor:technical-infinite-t} and Corollary~\ref{cor:complement-lusd}, we obtain $(\LSD,\USD)({\bf w},{\bf t}) \in (\{ [\frac{1}{4},\frac{3}{4}], [1 - \frac{3}{4}, 1 - \frac{1}{4}] \}, \{ [\frac{1}{4},\frac{3}{4}], [1 - \frac{3}{4}, 1 - \frac{1}{4}] \}) = ([\frac{1}{4},\frac{3}{4}], [\frac{1}{4},\frac{3}{4}])$, as desired.
\item
From inspection of the Fife automaton for overlap-free infinite binary words, we see that ${\bf x} \in \{ {\tt 0}^n\{{\tt 2}({\tt 31})^\omega,{\tt 4}({\tt 13})^\omega\} \st n \in \bbN \}$.
Note that $\FBE({\tt 0}^n{\tt 4}({\tt 13})^\omega) = \overline{\FBE({\tt 0}^n{\tt 2}({\tt 31})^\omega)}$.
Hence, by Proposition~\ref{prop:edge} and Corollary~\ref{cor:complement-lusd}, we obtain $(\LSD,\USD)({\bf w},{\bf t}) \in \{ (\frac{1}{3},\frac{2}{3}), (1-\frac{2}{3},1-\frac{1}{3}) \} = \{ (\frac{1}{3},\frac{2}{3}) \} \subset ([\frac{1}{4},\frac{3}{4}],[\frac{1}{4},\frac{3}{4}])$, as desired.
\item
From inspection of the Fife automaton for overlap-free infinite binary words, we see that ${\bf x} \in \{ {\tt 0}^n\{{\tt 2}({\tt 31})^{\frac{m}{2}},{\tt 4}({\tt 13})^{\frac{m}{2}}\}{\tt 0}\{{\tt 1}, {\tt 3}\}{\bf y} \st n,m \in \bbN, {\bf y} \in \{{\tt 0}, {\tt 1}, {\tt 3}\}^\omega \}$, whence
\begin{align*}
{\bf w} &\in \Sigma_2^\omega \intersect \left( \bigunion_{n,m\in\bbN} \Sigma_2^{2^{n+m+1}}\{t_{n+m+2},\overline{t_{n+m+2}}\} \prod_{k=n+m+3}^\infty \{\epsilon, t_k, \overline{t_k}\} \right) \\
&\subseteq \bigunion_{n,m\in\bbN} \Sigma_2^{2^{n+m+1}}\{t_{n+m+2},\overline{t_{n+m+2}}\}\{t_{n+m+3},\overline{t_{n+m+3}}\}^\omega,
\end{align*}
so there is a $k \in \bbN$ such that ${\bf w}[2^k \dotdot \infty] \in \{t_{k+1},\overline{t_{k+1}}\}\{t_{k+2},\overline{t_{k+2}}\}^\omega$.
By Observation~\ref{obsv:lusd-ignore-finite} and Corollary~\ref{cor:technical-infinite}, we obtain
\begin{align*}
(\LSD,\USD)({\bf t},{\bf w})
&= (\LSD,\USD)({\bf t}[2^{k+2} \dotdot \infty],{\bf w}[2^{k+2} \dotdot \infty]) \\
&= (\LSD,\USD)(\underbrace{{\bf t}[2^{k+2} \dotdot \infty]}_{\in \{t_{k+2},\overline{t_{k+2}}\}^\omega},(\underbrace{{\bf w}[3 \cdot 2^k \dotdot \infty]}_{\in \{t_{k+2},\overline{t_{k+2}}\}^\omega})[2^k \dotdot \infty]) \\
&\in \left(\left[\frac{1}{4},\frac{3}{4}\right],\left[\frac{1}{4},\frac{3}{4}\right]\right),
\end{align*}
as desired.
\item
From inspection of the Fife automaton for overlap-free infinite binary words, we see that ${\bf x} \in \{ {\tt 0}^n\{{\tt 1}, {\tt 3}\}{\tt 0}^m\{{\tt 1}, {\tt 3}\}{\bf y} \st n,m \in \bbN, {\bf y} \in \{{\tt 0}, {\tt 1}, {\tt 3}\}^\omega \}$, whence
\begin{align*}
{\bf w}
&\in \Sigma_2^\omega \intersect \left(\bigunion_{n,m\in\bbN} \{t_n,\overline{t_n}\}\{t_{n+m+1},\overline{t_{n+m+1}}\} \prod_{k=n+m+2}^\infty \{\epsilon,t_k,\overline{t_k}\} \right) \\
&\subseteq \bigunion_{n,m\in\bbN} \{t_n,\overline{t_n}\}\{t_{n+m+1},\overline{t_{n+m+1}}\} \{t_{n+m+2},\overline{t_{n+m+2}}\}^\omega,
\end{align*}
so there are $k,l \in \bbN$ such that ${\bf w} \in \{t_k,\overline{t_k}\}\{t_{k+l+1},\overline{t_{k+l+1}}\}\{t_{k+l+2},\overline{t_{k+l+2}}\}^\omega$.
By Observation~\ref{obsv:lusd-ignore-finite} and Corollary~\ref{cor:technical-infinite}, we obtain
\begin{align*}
(\LSD,\USD)({\bf t},{\bf w})
&= (\LSD,\USD)({\bf t}[2^{k+l+2} \dotdot \infty],{\bf w}[2^{k+l+2} \dotdot \infty]) \\
&= (\LSD,\USD)(\underbrace{{\bf t}[2^{k+l+2} \dotdot \infty]}_{\in \{t_{k+l+2},\overline{t_{k+l+2}}\}^\omega},(\underbrace{{\bf w}[2^k + 2^{k+l+1} \dotdot \infty]}_{\in \{t_{k+l+2},\overline{t_{k+l+2}}\}^\omega})[2^{k+l+1} - 2^k \dotdot \infty]) \\
&\in \left(\left[\frac{1}{4},\frac{3}{4}\right],\left[\frac{1}{4},\frac{3}{4}\right]\right),
\end{align*}
as desired. \qedhere
}
\end{enumerate}
\end{proof}

\section{Future work}

Using the Fife automaton for overlap-free infinite binary words, we computed similarity densities of long prefixes of all overlap-free infinite binary words (up to a certain length) with prefixes of ${\bf t}$. Inspection of the compuation results immediately suggests the following improvement to Theorem~\ref{thm:main}.

\begin{conj} \label{conj:main+}
For all overlap-free ${\bf w} \in \Sigma_2^\omega \delete \{{\bf t}, {\bf \overline{t}}\}$, we have
$\frac{1}{3} \leq \LSD({\bf w},{\bf t}) \leq \USD({\bf w},{\bf t}) \leq \frac{2}{3}$.
\end{conj}

Note that the bounds in Conjecture~\ref{conj:main+} are tight due to Proposition~\ref{prop:h}.
Computational evidence also suggests that these bounds are also tight for many other overlap-free infinite binary words.

However, Conjecture~\ref{conj:main+} cannot be proved just by using the technique we used to prove Theorem~\ref{thm:main}.
This is because the bounds in Lemma~\ref{lem:technical-one} (and, more transparently, Corollary~\ref{cor:technical-finite}) are tight.
For example, $\SD(t_2, t_3[1 \dotdot 4]) = \SD({\tt 0110}, {\tt 1101}) = \frac{1}{4}$.
More generally,
for any $n \in \bbN$, we have $\SD(t_{n+2}, t_{n+3}[2^n \dotdot 2^{n+2} + 2^n - 1]) = \frac{1}{4}$.

On the other hand, our proof of Theorem~\ref{thm:main} never used the overlap-free property directly; we merely used it indirectly via the Fife automaton.
As such, our proof of Theorem~\ref{thm:main} works for all images of $\FBE$ provided the argument to $\FBE$ is of the form required for one of the four cases presented in the proof, regardless of whether the resulting word is overlap-free.
Namely, we have the following more general, but much more cumbersome, theorem.

\begin{thm} \label{thm:main+}
For all $
{\bf x} \in
\{{\tt 1},{\tt 2},{\tt 3},{\tt 4}\} \Sigma_5^*{\tt 0}^\omega 
\union
{\tt 0}^* \{{\tt 2}({\tt 31}), {\tt 4}({\tt 13})\}^\omega $
\begin{displaymath}
\union\ 
{\tt 0}^* \{{\tt 2}({\tt 31})^*\{\epsilon,{\tt 3}\},{\tt 4}({\tt 13})^*\{\epsilon,{\tt 1}\}\}{\tt 0}\{{\tt 1},{\tt 3}\}\{{\tt 0},{\tt 1},{\tt 3}\}^\omega 
\union
({\tt 0}^*\{{\tt 1},{\tt 3}\})^2\{{\tt 0},{\tt 1},{\tt 3}\}^\omega
\end{displaymath}
and
\[
{\bf w} \in
\begin{cases}
\{\FBE({\bf x}, {\tt 0}),\FBE({\bf x}, {\tt 1})\}, & \text{if ${\bf x}$ ends in ${\tt 0}^\omega$}; \\
\{\FBE({\bf x})\}, &\text{otherwise},
\end{cases}
\]
we have
\[
\frac{1}{4} \leq \LSD({\bf w},{\bf t}) \leq \USD({\bf w},{\bf t}) \leq \frac{3}{4}.
\]
\end{thm}

Note that Theorem~\ref{thm:main+} is indeed more general than Theorem~\ref{thm:main}, since, for example, ${\tt 13}^\omega$ is not a valid path in the Fife automaton for overlap-free infinite binary words (indeed, $\FBE({\tt 13}^\omega)$ begins with the overlap ${\tt 01010}$) and $\FBE({\tt 13}^\omega)$ also is not just a shift of ${\bf t}$ or $\overline{\bf t}$, but Theorem~\ref{thm:main+} nevertheless implies that $\frac{1}{4} \leq \LSD(\FBE({\tt 13}^\omega),{\bf t}) \leq \USD(\FBE({\tt 13}^\omega),{\bf t}) \leq \frac{3}{4}$.

Together, Conjecture~\ref{conj:main+} and Theorem~\ref{thm:main+} suggest the following more general question.

\begin{quest}
For each $n \in \bbN \delete \{0,1\}$, $r,s \in [0,1]$, and ${\bf x} \in \Sigma_n^\omega$, let
\[
S_{n,r,s}({\bf x}) \coloneq \{ {\bf y} \in \Sigma_n^\omega \st r \leq \LSD({\bf x}, {\bf y}) \leq \USD({\bf x}, {\bf y}) \leq s \}.
\]

What are $S_{2,\frac{1}{4},\frac{3}{4}}({\bf t})$ and $S_{2,\frac{1}{3},\frac{2}{3}}({\bf t})$?
\end{quest}

Another avenue of investigation is to consider what makes ${\bf t}$ so special in the sense of Theorem~\ref{thm:main}.
As mentioned in the introduction, Theorem~\ref{thm:main} is false if we replace ${\bf t}$ with an arbitrary overlap-free infinite binary word.
However, perhaps there are specific words other than ${\bf t}$ and $\overline{\bf t}$ that do share similar properties.
In other words, we raise the following question.

\begin{quest}
Let $\calO$ denote the set of all overlap-free infinite binary words.

What is $\{ {\bf x} \in \Sigma_2^\omega \st \calO \subseteq S_{2,\frac{1}{4},\frac{3}{4}}({\bf x}) \}$?
What if we replace $\frac{1}{4},\frac{3}{4}$ with $\frac{1}{3},\frac{2}{3}$?
\end{quest}

A third avenue of investigation is to consider what occurs in words that avoid higher powers in place of being overlap-free (which are essentially $(2+\epsilon)$- or $2^+$-powers).
In fact, there is a Fife automaton characterizing $\frac{7}{3}$-power-free infinite binary words having the same encoding mechanism as the Fife automaton for overlap-free infinite binary words but with more states and different transitions \cite{Blondel&Cassaigne&Jungers:2009,Rampersad&Shallit&Shur:2011}. 
However, initial inspection of the automaton for $\frac{7}{3}$-power-free infinite binary words suggests that our proof of Theorem~\ref{thm:main} cannot be extended to account for all $\frac{7}{3}$-power-free infinite binary words because there are many more edges labeled ${\tt 2}$ and ${\tt 4}$ in the Fife automaton for $\frac{7}{3}$-power-free infinite binary words, resulting in valid paths that contain infinitely many ${\tt 2}$s and ${\tt 4}$s, but our proof of Theorem~\ref{thm:main} heavily relied on there being at most one occurrence of ${\tt 2}$ or ${\tt 4}$ (which must be preceeded by a string of ${\tt 0}$s if it occurs) in the path taken through the automaton so that the infinite binary word corresponding to the path eventually ``lags behind'' the prefixes $t_n$ of ${\bf t}$ in the sense that each successive $n$\superscript{th} symbol in the path can only generate positions prior to $2^n$.
Nevertheless, computational evidence suggests that Theorem~\ref{thm:main} and even Conjecture~\ref{conj:main+} can be generalized even further.

\begin{conj} \label{conj:main++}
For all $\frac{7}{3}$-power-free ${\bf w} \in \Sigma_2^\omega \delete \{{\bf t}, {\bf \overline{t}}\}$, $\frac{1}{3} \leq \LSD({\bf w},{\bf t}) \leq \USD({\bf w},{\bf t}) \leq \frac{2}{3}$.
\end{conj}

Finally, we revisit the notion, already mentioned in Remark~\ref{rem:density-generalization}, that $\LSD$ and $\USD$ are not new ideas, and not just in number theory.
In fact, $1 - \LSD$ is a pseudometric on $\Sigma^\bbN$, called the Besicovitch pseudometric, which has already been studied from the perspective of discrete dynamical systems such as \cite{Blanchard&Formenti&Kurka:1997}.
Also studied in \cite{Blanchard&Formenti&Kurka:1997} is the Weyl pseudometric, which suggests the following slightly different notion of similarity density, considering all blocks of a given size instead of just blocks from the beginning.
\begin{align*}
\LSD_{\text{Weyl}}({\bf x},{\bf y}) &= \liminf_{n\to\infty} \ \inf_{k\in\bbN} \SD({\bf x}[k\dotdot k+n-1],{\bf y}[k\dotdot k+n-1]), \\
\USD_{\text{Weyl}}({\bf x},{\bf y}) &= \limsup_{n\to\infty} \ \sup_{k\in\bbN} \SD({\bf x}[k\dotdot k+n-1],{\bf y}[k\dotdot k+n-1]).
\end{align*}
With this notion of Weyl similarity density, analogous to the Besicovitch case, we have that $1 - \LSD_{\text{Weyl}}$ is the Weyl pseudometric.
The Besicovitch and Weyl pseudometrics share some topological properties, but
the Besicovitch pseudometric is complete while the Weyl pseudometric is not \cite{Blanchard&Formenti&Kurka:1997}.
This fact suggests one might be able to shed further light on some of
the questions above by also considering the Weyl similarity density; perhaps
several different
notions of similarity density, when taken together, can characterize
the overlap-free infinite binary words.

\bigskip
\noindent{\it Acknowledgments.}  We are grateful to Chao Hsien Lin for having
suggested the question we study here.  We thank the referees for a careful reading,
and Joel Ouaknine and Stefan Kiefer for having pointed out the 
paper \cite{Blanchard&Formenti&Kurka:1997}.

%\nocite{*}
\bibliographystyle{eptcs}
\bibliography{abbrevs,du9}

\end{document}